\documentclass{amsproc}

\usepackage[T2A]{fontenc}
\usepackage[cp1251]{inputenc}
\usepackage[english]{babel}

\usepackage{amsfonts}
\usepackage{amsthm}
\usepackage{amssymb}
\usepackage{amsmath}
\usepackage{amsopn}
\usepackage{bm}
\usepackage{array}
\usepackage{xcolor}
\usepackage{mathdots}
\usepackage{graphicx}
\usepackage{enumitem}

\newtheorem{prop}{Proposition}
\newtheorem{theo}{Theorem}
\newtheorem{lem}{Lemma}

\newtheorem{conj}{Conjecture}

\theoremstyle{definition}
\newtheorem{defin}{Definition}

\newtheorem{rem}{Remark}

\usepackage{cite}
\allowdisplaybreaks[4]

\DeclareMathOperator{\Ad}{Ad}

\DeclareMathOperator{\tr}{tr}

\newcommand{\x}{\mathrm{x}}
\newcommand{\rmx}{\mathrm{x}}
\newcommand{\rmt}{\mathrm{t}}
\newcommand{\rmw}{\mathrm{w}}
\newcommand{\rmv}{\mathrm{v}}

\newcommand{\rmd}{\mathrm{d}}

\newcommand{\M}{\mathcal{M}}

\newcommand{\X}{\textsf{X}}
\newcommand{\Y}{\textsf{Y}}

\newcommand{\E}{\textsf{E}}

\renewcommand{\r}{\mathsf{r}}
\newcommand{\s}{\mathsf{s}}

\newcommand{\Lmatr}{\mathsf{L}}

\DeclareMathOperator{\Complex}{\mathbb{C}}
\DeclareMathOperator{\Real}{\mathbb{R}}
\DeclareMathOperator{\Natural}{\mathbb{N}}
\DeclareMathOperator{\Integer}{\mathbb{Z}}

\DeclareMathOperator{\Ibb}{\mathbb{I}}
\DeclareMathOperator{\Jac}{Jac}
\DeclareMathOperator{\wgt}{wgt}
\DeclareMathOperator*{\res}{res}

\DeclareMathOperator{\Discr}{Discr}
\DeclareMathOperator{\Ord}{Ord}

\newcommand{\wFr}{\mathfrak{w}}

\DeclareMathOperator{\ReN}{\mathrm{Re}}
\DeclareMathOperator{\ImN}{\mathrm{Im}}

\title{Algebro-geometric integration of the Boussinesq hierarchy}
\author{Julia Bernatska}
\address{University of Connecticut, Department of Mathematics,
341 Mansfield Rd, Storrs, CT 06269}
\email{julia.bernatska@uconn.edu}

\author{Taras Skrypnyk}
\address{Bogolyubov Institute for Theoretical Physics, Metroligichna st. 14b, Kyiv 03143, Ukraine}
\email{tskrypnyk@bitp.kiev.ua}

\begin{document}

\maketitle

\begin{abstract}
We construct an integrable  hierarchy of the Boussinesq equation using  
the Lie-algebraic  approach of Holod-Flashka-Newell-Ratiu.
We show that finite-gap hamiltonian systems of the hierarchy arise on coadjoint orbits
in the loop algebra of $\mathfrak{sl}(3)$, and possess spectral curves from
the family of $(3,3N\,{+}\,1)$-curves, $N\,{\in}\, \Natural$.
Separation of variables leads to the Jacobi inversion problem on the mentioned curves,
which is solved in terms of the corresponding multiply periodic functions.
An exact finite-gap solution of the Boussinesq equation is obtained explicitly,
and a conjecture on the reality conditions is made.
The obtained solutions are computed for several spectral curves, and illustrated graphically.

\end{abstract}



\section{Introduction}

According to \cite[Sect.10.2.1]{Pol2004}, the canonical  Boussinesq equation is
\begin{equation}\label{BousEqCanon}
\rmw_{\rmt\rmt} + \rmw \rmw_{\rmx\rmx} + \rmw_{\rmx}^2 + \rmw_{\rmx\rmx\rmx\rmx} = 0.
\end{equation}
At the same time, the original equation for shallow-water waves propagating in both directions,
obtained by Boussinesq (1872),
has the form
\begin{equation}\label{BousEqOrig}
\rmw_{\rmt\rmt} - g h \rmw_{\rmx\rmx}  - g h (\tfrac{3}{2} h^{-1} \rmw^2
+ \tfrac{1}{3} h^2  \rmw_{\rmx\rmx}  )_{\rmx\rmx} = 0,
\end{equation}
where  $\rmw(\rmx,\rmt)$ is the free surface elevation,
$h$ denotes the water depth, and $g$ is the  gravitational acceleration.
The normalized form is
\begin{equation}\label{BousEqNorm}
\rmw_{\rmt\rmt} - \rmw_{\rmx\rmx}  - (3 \rmw^2  +  \rmw_{\rmx\rmx}  )_{\rmx\rmx} = 0.
\end{equation}
In \cite{BZ2002}, four different cases of the Boussinesq equation are considered
\begin{equation}\label{BousEq4Types}
\tfrac{3}{4} a^2 \rmw_{\rmt\rmt} - b \rmw_{\rmx\rmx} +
 ( \tfrac{1}{4} \rmw_{\rmx\rmx} + \tfrac{3}{2} \rmw^2 )_{\rmx\rmx} = 0.
\end{equation}
with $b \,{=}\, {\pm} 1$, $a^2 \,{=}\, {\pm} 1$.
The properties of solutions of the  equation
depend on a choice of the two signs.

We briefly recall the progress in solving the Boussinesq equation.
In \cite{Hir1973}, $N$-soliton solutions to \eqref{BousEqNorm} is obtained
by using the ansatz for soliton solutions to the KdV equation.
At the same time, the Lax pair for the Boussinesq equation was
constructed in \cite{Zakh1973}, which proves its complete integrability, and applicability
of the inverse scattering method. The Boussinesq equation \eqref{BousEqCanon}
possesses the Painlev\'{e} property, as shown in \cite{Weiss1983},
where also the associated B\"{a}cklund transform is found.
The techniques of inverse scattering theory is used to construct solutions to the Boussinesq equation in
\cite{DTT1982}.

As shown in \cite{FST1983}, waves of sufficiently large amplitude reveal instability resulting in collapse.
A systematic study of the solitonic sector of the Boussinesq equation is presented in \cite{BZ2002},
and formation of singularity in the process of two-soliton interaction is shown.
An application of the dressing method to the Boussinesq equation is also
presented in \cite[Sect.\,6]{BZ2002}.

In \cite{Kap1998}, trigonometric and breezer-type solutions 
to \eqref{BousEqNorm} are found in $2$- and $3$-soliton sectors.
In \cite{MS1987}, solutions  to \eqref{BousEqCanon} 
associated with a genus three non-hyperelliptic curve are obtained.
The curve under consideration has a split Jacobian variety, and the obtained solution
is expressed in terms of two one-dimensional theta functions.

In \cite{McK1981}, the methods of algebraic geometry shown in \cite{Kri1977}
are applied with the goal to construct solutions to the equation
$$ \rmw_{\rmt\rmt}  - (\tfrac{4}{3} \rmw^2  - \tfrac{1}{3} \rmw_{\rmx\rmx}  )_{\rmx\rmx} = 0. $$
The Boussinesq equation is associated with a three-sheeted curve, in general of infinite genus.
In more detail, the Boussinesq equation arises in 
a hamiltonian system, points of which map to non-special divisors of the curve.
And the Boussinesq flow produces a flow on the Jacobian variety, which is a straight line.

In the present paper,  we construct the hierarchy of the equation
\begin{equation*}
3 \rmw_{\rmt\rmt} + 4 \rmw \rmw_{\rmx\rmx} + 4 \rmw_{\rmx}^2 + \rmw_{\rmx\rmx\rmx\rmx} = 0,
\end{equation*}
which coincides with the canonical form, up to rescaling.
The hierarchy is derived by means of the orbit method after Holod  
\cite{HolMKdV1982} and Flashka-Newell-Ratiu \cite{FNR1983},
and completely integrated by methods of algebraic geometry.
The finite-gap solution is expressed in terms of the Kleinian function $\wp_{1,1}$
associated with trigonal curves.
The same function associated with a family of hyperelliptic curves serves as 
 the finite-gap solution to the Korteweg---de Vries equation, see 
 \cite[Theorem\,4.12]{belHKF}, and \cite{BerKdV2024}.
This is the first time, when quasi-periodic solution to
the Boussinesq equation is given explicitly, and illustrated by plots.

\section{Outline of methods}
\subsection{Integrable hierarchies in 1+1 dimensions}
According to \cite{ZakhShab1972}, the inverse scattering method is applicable to equations
which admit the Lax representation
$$ \partial_{\rmt} L = [A,L],$$
where $\partial_{\rmt} \equiv \partial/\partial \rmt$, and $L$, $A$ are differential operators
which act on a function $\rmw(\rmx,\rmt)$ of the spatial variable $\rmx$, and the time variable $\rmt$.
Such equations are called
completely integrable, or soliton equations.

 Alternatively,  the Lax representation can be rewritten in the zero-curvature form
$$ \partial_{\rmt} A_{\text{st}} - \partial_{\rmx} A_{\text{ev}} + [A_{\text{st}},A_{\text{ev}}] = 0,$$
for the matrix-valued functions $A_{\text{st}}$, $A_{\text{ev}}$,  constructed from the Lax pair $L$, $A$.

In \cite{FNR1983},  a relationship of the eigenvalue problem for $L$
with the loop $\mathfrak{sl}(2)$ algebra of formal series is established.
Soliton equations are associated with $\mathfrak{sl}(2)$
eigenvalue problems for polynomials in a spectral parameter.
It is shown that (i) <<the conserved densities and fluxes of the usual ANKS hierarchy are identified
with conserved densities and fluxes for the polynomial eigenvalue problems>>; and
(ii) <<the hamiltonian structure of soliton equations
arise from the Kostant---Kirillov symplectic
structure on a coadjoint orbit in an infinite-dimensional Lie algebra>>.

Independently, a variety of completely integrable hierarchies
are constructed on coadjoint orbits in the loop, and elliptic algebras of $\mathfrak{sl}(2)$,
see \cite{HolMKdV1982,HolNEq1984,HolALL1984,Hol1ALL987,HolKS1984,HolS1987},
and later, hierarchies on coadjoint orbits in the loop algebra of $\mathfrak{sl}(3)$,
see \cite{HolKK1997,HolK1999}. The common approach used in these papers 
was called the orbit method.
Recently, some of the mentioned hierarchies have been  integrated 
by the methods of algebraic geometry in \cite{BerKdV2024,BerSinG2025,BerMKdV2025}.

As a development of the orbit method,
a wider class of integrable systems is derived based on the $\r$-matrix technique, see
\cite{SkrJPA2007, SkrJMP2013, SkrJMP2014, SkrJGP2014, SkrJGP2018, SkrJMP2021}

\subsection{A modern version of the orbit method}
Let $\mathfrak{g}$ be a  simple (or reductive) finite-dimensional Lie algebra
with the basis $\{\X_{a} \mid a=\overline{1,\dim \mathfrak{g}}\}$, and the commutation relations
\begin{equation}\label{comrel1}
[ \X_{a}, \X_{b}]=  \sum\limits_{c=1}^{\dim \mathfrak{g}}
C_{a b c} X_{c}.
\end{equation}

An infinite-dimensional Lie algebra  $\widetilde{\mathfrak{g}}$
is an algebra of formal power series in a spectral parameter $z$, see \cite{FNR1983,HolMKdV1982}, or more generally,
 an algebra  of meromorphic
$\mathfrak{g}$-valued functions of  $z$, see  \cite{ {SkrJMP2013,SkrJMP2014,SkrJGP2018}}.

The algebra $\widetilde{\mathfrak{g}}$ is graded by means of a grading operator $\mathfrak{d}$
with respect to $z$. Let $\{\widetilde{\X}_{a;\ell}(z) \mid a \in \overline{1,\dim \mathfrak{g}}$, $\ell\,{\geqslant}\,0\}$
be a basis of $\widetilde{\mathfrak{g}}$, where $\ell$ is an eigenvalue of the grading operator $\mathfrak{d}$.
In the case of homogeneous grading, the basis of $\widetilde{\mathfrak{g}}$ has the form
\begin{equation}\label{fbasis}
\widetilde{\X}_{a;\ell}(z) = \X_{a} z^\ell, \quad  \ell \geqslant 0,\ \ a \in \overline{1,\dim \mathfrak{g}}.
\end{equation}
The commutation relations respect grading
\begin{equation}\label{LoopAlgComRels}
[\widetilde{\X}_{a;\ell}(z), \widetilde{\X}_{b;\ell'}(z)] = \sum_{c=1}^{\dim \mathfrak{g}}
C_{abc}  \widetilde{\X}_{c;\ell+\ell'}(z),\quad
a, b \in \overline{1,\dim \mathfrak{g}},\ \ \ell,\, \ell' \geqslant 0.
\end{equation}

Let $\widetilde{\mathfrak{g}}^\ast$  be the dual space of $\widetilde{\mathfrak{g}}$ with respect to
a bilinear form $\langle \cdot,\cdot \rangle$, and
$\{\widetilde{\X}^\ast_{a;\ell}(z) \mid a \in \overline{1,\dim \mathfrak{g}}, \ \ell \geqslant 0 \}$
denote the basis of $\widetilde{\mathfrak{g}}^\ast$ such that
$\langle \widetilde{\X}_{a;\ell}(z),\widetilde{\X}^\ast_{b;\ell'}(z) \rangle = \delta_{a,b} \delta_{\ell,\ell'}.$

 Let $\mathcal{M}_N \,{\subset}\, \widetilde{\mathfrak{g}}^\ast$, $N\,{\in}\,\Natural$,
have a generic element  of the form
 \begin{equation}\label{LMatrGen}
 \Lmatr_N(z)= \s(z) + \sum_{\ell = 0}^{N-1} \sum_{a=1}^{\dim \mathfrak{g}} L_{a;\ell} \widetilde{\X}^\ast_{a;\ell}(z).
 \end{equation}
The coordinates  $\{L_{a;\ell} \mid a \in \overline{1,\dim \mathfrak{g}}, \ \ell \,{\in}\,\overline{0,N-1}\}$ 
 of $\mathcal{M}_N$ serve as dynamic variables in a finite-gap hamiltonian system,
and a shift element $\s(z)$, see below for more details, does not depend of dynamic variables.
The manifold $\mathcal{M}_N$ is equipped with the Lie---Poisson bracket
\begin{equation}\label{LPbraGen}
\{L_{a;\ell}, L_{b;\ell'}\} = \sum_{c=1}^{\dim \mathfrak{g}}
C_{abc}  L_{c;\ell+\ell'},\quad
 a,b \in \overline{1,\dim \mathfrak{g}},\ \  \ell, \ell' \in \overline{0,N-1}.
\end{equation}

For each $N$ there exists a finite amount of integrals of motion, which are in involution with respect to
 \eqref{LPbraGen}. Some of the integrals of motion annihilate the Lie---Poisson bracket, 
i.e. are the Casimir functions. They  impose constraints, and fix a coadjoint orbit of 
the corresponding loop group. The coadjoint orbits
serve as  phase spaces of finite-gap hamiltonian systems. The remaining integrals of motion
give rise to hamiltonian flows, and so called hamiltonians.
The number of hamiltonians coincides with the number of degrees of freedom 
in the case of a good choice of the shift element.

Among all hamiltonians two are chosen, say $h_{\text{st}}$, $h_{\text{ev}}$,
which generate the flows
$$
\partial_{\rmx} L_{a;\ell} = \{h_{\text{st}}, L_{a;\ell}\}, \qquad
\partial_{\rmt} L_{a;\ell} = \{h_{\text{ev}}, L_{a;\ell}\},
$$
or in the matrix form
\begin{equation}\label{MatrFlow}
\partial_{\rmx} L_{a;\ell} = [ \nabla h_{\text{st}}, L_{a;\ell}], \qquad
\partial_{\rmt} L_{a;\ell} = [ \nabla h_{\text{ev}}, L_{a;\ell}]
\end{equation}
where
$$
\nabla h = \sum_{\ell = 0}^{N-1} \sum_{a=1}^{\dim \mathfrak{g}}
\frac{\partial h}{\partial L_{a;\ell}} \widetilde{\X}_{a;\ell}(z).
$$
The compatibility condition of \eqref{MatrFlow} gives the zero-curvature representation
of the soliton equation in question, namely
\begin{equation}\label{zc}
\partial_{\rmt} \nabla h_{\text{st}} -
\partial_{\rmx} \nabla h_{\text{ev}}
+ [\nabla h_{\text{st}},\nabla h_{\text{ev}}] = 0.
\end{equation}

Each finite-gap solution of a soliton equation is represented by a trajectory
of a hamiltonian system from the corresponding hierarchy.
Therefore, initial conditions are introduced by fixing an orbit which serves
as the phase space of the chosen  hamiltonian system,
and fixing values of hamiltonians.

\subsection{Classical $r$-matrices}
We will construct an in\-fi\-ni\-te-di\-men\-sio\-nal Lie algebra $\widetilde{\mathfrak{g}}$
and the dual space $\widetilde{\mathfrak{g}}^\ast$ by means of the $\r$-matrix
associated with an integrable hierarchy.

\begin{defin}
A function
$\r: \Complex^2 \to \mathfrak{g}\otimes \mathfrak{g}$ of the form
$$ \r(z,\zeta)\equiv
\sum_{a,b=1}^{\dim \mathfrak{g}}
r_{ab}(z,\zeta)\, \X_{a} \otimes \X_{b},$$
where $r_{ab}(z,\zeta)$ are scalar functions of $z$, $\zeta\,{\in}\,\Complex$,  
is called an \emph{$\r$-matrix} 
if $\r$ satisfies  the permuted Yang---Baxter
equation \cite{Avan}, also known as the generalized classical Yang---Baxter  
equation, see \cite{ {SkrJMP2013,SkrJMP2014,SkrJMP2018,SkrJMP2021}},
\begin{equation}\label{GCYB}
[\hat{\r}^{12}(z_1,z_2), \hat{\r}^{13}(z_1,z_3)] = [\hat{\r}^{23}(z_2,z_3),
\hat{\r}^{12}(z_1,z_2)]-[\hat{\r}^{32}(z_3,z_2), \hat{\r}^{13}(z_1,z_3)],
\end{equation}
where $z_1$, $z_2$, $z_3 \in \Complex$, and
$\hat{\r}^{12}(z_1,z_2)\equiv \textstyle
\sum_{a,b=1}^{\dim \mathfrak{g}} r_{ab}(z_1,z_2) \X_{a} \otimes \X_{b}\otimes \Ibb$ etc.
\end{defin}

A non-degenerate $\r$-matrix possesses the decomposition
\begin{equation}\label{rMatrDcmps}
\r(z,\zeta) = \r_0(z,\zeta) + \sum\limits_{a,b=1}^{\mathrm{dim}
\mathfrak{g}} \frac{\tilde{\mu}_{ab} }{z-\zeta}   \X_{a}\otimes \X_{b},
\end{equation}
where $\r_0: \Complex^2 \to \mathfrak{g}\otimes \mathfrak{g}$ is the regular part of $\r$, 
and $(\tilde{\mu}_{ab}) \,{=}\, \mu^{-1}$ such that 
$\mu \,{=}\, (\mu_{ab})$, $\mu_{ab} \,{=}\, \langle \X_{a}, \X_{b} \rangle$, where 
$\langle\cdot ,\cdot \rangle$ denotes a bilinear  form on~$\mathfrak{g}$.

Let ${\mathfrak{g}}(u^{-1},u)$ be the Lie algebra of the formal Laurent series.
A non-degenerate $\r$-matrix is associated   with  two infinite-dimensional subalgebras of
${\mathfrak{g}}(u^{-1},u)$,
namely
 \begin{equation}
\mathfrak{g}(z^{-1},z)= \mathfrak{g}(z) \oplus  \widetilde{\mathfrak{g}}^-_{\r},
\end{equation}
where ${\mathfrak{g}}(z) \,{=}\, \widetilde{\mathfrak{g}}$
denotes the algebra of $\mathfrak{g}$-valued power series, and  $\widetilde{\mathfrak{g}}^-_{\r}$ is
an algebra of  meromorphic functions defined by the $\r$-matrix, see \cite{SkrJMP2013}.
The dual space $(\widetilde{\mathfrak{g}}^-_{\r})^*$ is not employed in the present paper.

A basis of the dual space $\widetilde{\mathfrak{g}}^\ast$
is obtained by the formula, see  \cite{SkrJMP2013},
\begin{equation}\label{fbasis}
\widetilde{\X}^\ast_{a;\ell}(z) = \frac{1}{\ell!}
 \sum_{b=1}^{\dim \mathfrak{g}}
\big( \partial^{\ell}_{\zeta}  r_{ba}(\zeta,z)|_{\zeta=0} \big) \X_{b}, \quad \ell\geqslant 0,\ \
a \in \overline{1,\dim \mathfrak{g}},
\end{equation}
where functions $r_{ba}(\zeta,z)$ define the $\r$-matrix $\r(\zeta,z)$.
A generic element of $\widetilde{\mathfrak{g}}^\ast$ has the form
$$
\Lmatr^+(z) = \sum_{\ell \geqslant 0} \sum_{a=1}^{\dim \mathfrak{g}}
L_{a;\ell} \widetilde{\X}^\ast_{a;\ell}(z).
$$
The generalized Yang-Baxter equation implies, see \cite{SkrJMP2013},
\begin{equation}\label{rmbrp}
\{\Lmatr^+(z) \overset{\otimes}{,} \Lmatr^+(\zeta)\} = [\r(z,\zeta),\Lmatr^+(z) \otimes \Ibb] -
[\r(\zeta,z), \Ibb \otimes \Lmatr^+(\zeta)].
\end{equation}

\begin{defin}[\cite{SkrJPA2007}]
A function of the form $\s(z) \,{=}\, \textstyle\sum_{a=1}^{\dim \mathfrak{g}} s_{a}(z) \X_{a}$
such that
\begin{equation}\label{sheq}
[\r(z,\zeta),\s(z)\otimes \Ibb]- [\r(\zeta,z),\Ibb \otimes \s(\zeta)]=0
\end{equation}
is called a \emph{shift element with respect to the $\r$-matrix} $\r(z,\zeta)$.
\end{defin}

\begin{prop}[\cite{SkrJPA2007}]
Let $\Lmatr(z) = \Lmatr^{+}(z) + \s(z)$, then $\Lmatr(z)$
satisfies  \eqref{rmbrp}.
\end{prop}

\begin{theo}[\cite{BV}]
If $\Lmatr(z)$ satisfies  \eqref{rmbrp}, then
$$\{\tr \Lmatr(z)^k, \tr \Lmatr(\zeta)^n\}=0, \quad \forall k,n\geqslant 0.$$
\end{theo}

\begin{prop}\label{polr0}
Let the regular part $r_0(z,\zeta)$ of an $\r$-matrix be a polynomial in $z$ and a Laurent polynomial in $\zeta$.
Let a shift element $\s(z)$ be a Laurent polynomial in $z$. Then the integrals
\begin{equation}\label{integrp}
I^{+}_{k,n} = \res_{z=0} z^{-n-1} \tr \Lmatr(z)^k
\end{equation}
are  polynomials in  $L_{a;\ell}$, $a \in \overline{1,\dim \mathfrak{g}}$,
\ $l\geqslant 0$.
\end{prop}

\begin{rem}
If an $\r$-matrix satisfies the conditions of Proposition~\ref{polr0}, then
the algebra of Laurent series ${\mathfrak{g}}(z^{-1},z)$
can be replaced with the algebra of Laurent polynomials $\mathfrak{g}[z^{-1},z]$,
and the algebra of power series $\mathfrak{g}(z)$ with the algebra of Taylor polynomials ${\mathfrak{g}}[z]$.
If an $\r$-matrix does not satisfy the mentioned conditions, obtaining an integrable hierarchy
 from this $\r$-matrix  requires consideration of
 quasi-graded structure  in $\widetilde{\mathfrak{g}}$, see \cite{SkrJMP2013}.
\end{rem}

The dual space $\widetilde{\mathfrak{g}}^\ast$ equipped with the Poisson bracket  \eqref{rmbrp}
possesses an infinite sequence of  embedded ideals of  finite co-dimensions.
Cosets of these ideals in $\widetilde{\mathfrak{g}}^\ast$ form an infinite sequence
of  matrices
$$
\Lmatr^{+}_N(z) = \sum_{\ell=0}^{N-1} \sum\limits_{a=1}^{\dim \mathfrak{g}} 
L_{a;\ell} \widetilde{\X}^\ast_{a;\ell}(z),
\quad N \in\Natural.
$$
With a proper choice of the shift element $\s(z)$, 
the Lax matrix $\Lmatr_N (z) = \Lmatr^{+}_N(z) + \s(z)$ is obtained, which produces
a completely integrable hierarchy.

\subsection{Separation of variables}
We use the standard  variables of separation,
which are coordinates of a non-special divisor of the corresponding spectral curve,
suggested in  \cite{SklSep1}. Below, these variables of separation are obtained by means of the orbit method,
see \cite{BerHol07}, which simultaneously leads to a solution of the Jacobi inversion problem. The latter is
 the key part of algebro-geometric integration. 
The quasi-canonical property of the obtained variables of separation is proven
with the help of the pair of functions $\mathcal{A}$, $\mathcal{B}$, firstly introduced in \cite{SklSep1}.

The conditions on an $\r$-matrix which guarantee that the pair $\mathcal{A}$, $\mathcal{B}$ 
generates  quasi-canonical variables are obtained in
\cite{SkrJMP2018}. Separation of variables for a wide class of integrable systems derived from
$\mathfrak{gl}(n) \otimes \mathfrak{gl}(n)$-valued $\r$-matrices are  presented in \cite{SkrJMP2021},
as well as  the general Abel-type equations for separated variables associated with the spectral curve of any 
$\mathfrak{gl}(n)$-valued Lax matrix.

\section{The Boussinesq  hierarchy}
\subsection{Non-standard rational $r$-matrix}
Let $\mathfrak{g}\,{=}\,\mathfrak{gl}(3)$  with the standard basis elements
$\E_{ij}$, $i,j\in \{1,2,3\}$,  which obey the commutation relations
$$
[\E_{ij},\E_{kl}]=\delta_{kj} \E_{il} - \delta_{il} \E_{kj}, \quad i,j,k,l \in  \overline{1,3},
$$
and the standard bilinear form $\langle \X, \Y \rangle = \tr \X \Y$.

Let an $\r$-matrix for the  Boussinesq  hierarchy be defined by
\begin{equation}\label{rmbous}
\r(z,\zeta)=\sum\limits_{i,j=1}^3\frac{ \E_{ij}\otimes \E_{ji}}{z-\zeta}
+ \E_{31}\otimes(\E_{21}-2 \E_{32})-(\E_{21}+\E_{32})\otimes \E_{31}.
\end{equation}
This is a non-degenerate $\r$-matrix possessing the shift element
\begin{equation} \label{shelbous}
\s(z) = z \E_{31}+\E_{12}+\E_{23}.
\end{equation}

Indeed, by direct calculations the following is proven
\begin{prop}
The tensor (\ref{rmbous}) is a  classical $\r$-matrix, 
that is, (\ref{rmbous}) satisfies the generalized Yang-Baxter  equation (\ref{GCYB}). The  element
$\s(z)$ given by  (\ref{shelbous}) is a shift element, i.e. satisfies  equation  (\ref{sheq}).
\end{prop}

Note, that the $r$-matrix (\ref{rmbous}) in the context of the Boussinesq equation arose in \cite{RuGu}. 
In the present paper, we employ it  to derive the $\Lmatr$-matrix for the Boussinesq hierarchy, 
which we construct by means of the orbit method.



\subsection{Loop Lie algebra and dual space}
In what follows, $\mathfrak{g} \,{=}\, \mathfrak{sl}(3)$, and
$\widetilde{\mathfrak{g}}\,{=}\,\mathfrak{sl}(3)(z)$ is the loop
algebra of $\mathfrak{sl}(3)$ with the homogeneous grading. Thus, a basis of $\widetilde{\mathfrak{g}}$
has the form
\begin{equation}\label{fbasis}
\widetilde{\X}_{a;\ell} (z) = \X_a z^\ell, \quad   a \in \overline{1,8},\ \ \ell \geqslant 0,
\end{equation}
where $\X_a$ form the standard  basis in $\mathfrak{sl}(3)$, namely
\begin{equation}\label{sl3basis}
\begin{split}
\X_1 = \tfrac{1}{3} (2 \E_{11} - &\,\E_{22} - \E_{33}),\quad
\X_2 = \tfrac{1}{3} (\E_{11} + \E_{22} - 2\E_{33}),\quad
\X_3 = \E_{21},\\
\X_4 = \E_{32},\quad
&\X_5 = \E_{31},\quad
\X_6 = \E_{12}, \quad
\X_7 = \E_{23}, \quad
\X_8 = \E_{13}.
\end{split}
\end{equation}

A basis in the dual algebra $\mathfrak{sl}(3)^\ast$ is
\begin{equation}\label{sl3basis}
\begin{split}
\X^\ast_1 = \E_{11} - \E_{22},\quad  &\X^\ast_2 = \E_{22} - \E_{33},\quad \X^\ast_3 = \E_{12}\\
\X^\ast_4 = \E_{23},\quad \X^\ast_5 = \E_{13},\quad  &\X^\ast_6 = \E_{21},\quad
\X^\ast_7 = \E_{32},\quad \X^\ast_8 = \E_{31},
\end{split}
\end{equation}
and a basis in the dual space $\widetilde{\mathfrak{g}}^\ast$
constructed  from the $\r$-matrix \eqref{rmbous}  has the form
\begin{equation}\label{sl3basis}
\widetilde{\X}^\ast_{a;\ell} = z^{-\ell-1} \X^\ast_a + 
\big(\delta_{a,5} (\E_{21}-2\E_{32}) - \delta_{a,3} \E_{31}
- \delta_{a,4} \E_{31}\big)\delta_{\ell,0},\ \ 
  a \in \overline{1,8},\ \ \ell \geqslant 0.
\end{equation}

\subsection{Phase space of a $6N$-gap hamiltonian system}
For every $N\,{\in}\, \Natural$ hamiltonian systems of the Boussinesq hierarchy
can be constructed.
The phase space of such a system belongs to the manifold
\begin{gather}\label{BousMNSp}
\mathcal{M}_N = \Big\{  \Lmatr_N(z) =
\s(z) +  \sum_{\ell=0}^{N-1}  \sum_{a=1}^{8}  L_{a;\ell} \X_{a;\ell}^\ast\Big\},
\end{gather}
with coordinates $L_{a;m}$,  called  \emph{dynamic variables}.
Recall, that $\s(z)$ is defined by \eqref{shelbous}.
Due to the algebraic construction, we have
$$L_{a;\ell} = \langle \Lmatr_N(z), \X_{a;\ell} \rangle,\quad a\in \overline{1,8},\ \  \ell \in \overline{0,N-1}.$$
Let $\alpha_{i;m}$, $m\,{\in}\, \overline{1,N}$, be coordinates corresponding to Cartan elements of $\widetilde{\mathfrak{g}}$,
$\beta_{i;m}$ correspond to positive roots, and $\gamma_{i;m}$ to negative roots.
Actually,
\begin{equation}\label{DynVars}
\{L_{a;m-1} \mid a \in \overline{1,8}\} =
\{\alpha_{1;m},\, \alpha_{2;m},\, \beta_{1;m},\, \beta_{2;m},\, \beta_{3;m},\,
\gamma_{1;m},\, \gamma_{2;m},\, \gamma_{3;m}\}.
\end{equation}

The manifold $\mathcal{M}_N$ is equipped with the symplectic structure
given by the Lie-Poisson bracket: $\forall \mathcal{F},\, \mathcal{H}\in \mathcal{C}(\M_N)$
\begin{gather}\label{LiePoiBraBous}
\{\mathcal{F},\mathcal{H}\} = \sum_{\ell,\ell' =0}^{N-1}  \sum_{a, b=1}^8
W_{\ell,\ell'}^{a,b} \frac{\partial \mathcal{F}}{\partial L_{a;\ell}} \frac{\partial \mathcal{H}}{\partial L_{b;\ell'}},\quad
W_{\ell,\ell'}^{a,b}  = \langle \Lmatr_N(z), [\X_{a;\ell},\X_{b;\ell'}] \rangle.
\end{gather}

The action of the loop group $\widetilde{G} = \exp (\widetilde{\mathfrak{g}})$ splits $\M_N$ into orbits
$$\mathcal{O} = \{ \Lmatr_N(z) = \Ad^\ast_{g} \Lmatr^{\text{in}}_N(z) \mid g\in \widetilde{G}\},
\qquad \Lmatr^{\text{in}}_N(z)  \in \M_N.$$
Initial elements $\Lmatr^{\text{in}}_N(z) $ are taken from the Weyl chamber of $\widetilde{G}$ in $\M_N$.
The Weyl chamber is spanned by $\X_{1;\ell}^\ast$, $\X_{2;\ell}^\ast$, $\ell \in \overline{0,N-1}$, which are
 diagonal matrices. Each orbit serves as the phase space of a hamiltonian system in $\M_N$, as we see below.

Instead of $\Lmatr_N(z)$, we will work with
the polynomial matrix $\widetilde{\Lmatr}_N(z) = z^N \Lmatr_N(z)$, which has the form
\begin{equation}\label{BoussPhSp}
\begin{split}
& \widetilde{\Lmatr}_N(z) \,{=}\,
\begin{pmatrix} \alpha_1(z) & z^N \,{+}\, \beta_1(z) &  \beta_3(z) \\
\beta_{3;1} z^N + \gamma_1(z) & \alpha_2(z) - \alpha_1(z) & z^N + \beta_2(z) \\
z^{N+1} \,{-}\, (\beta_{1;1} \,{+}\, \beta_{2;1}) z^N \,{+}\, \gamma_3(z) &
{-}2\beta_{3;1} z^N \,{+}\, \gamma_2(z) & -\alpha_2(z)
\end{pmatrix},\\
&\alpha_i(z) = \sum_{m=1}^{N} \alpha_{i;m} z^{N-m},\quad i \in \{1,2\},\\
&\beta_i(z) = \sum_{m=1}^{N} \beta_{i;m} z^{N-m},\qquad
\gamma_i(z) = \sum_{m=1}^{N} \gamma_{i;m} z^{N-m},\quad i \in \{1,2,3\}.
\end{split}
\end{equation}
On the other hand,
$$\widetilde{\Lmatr}_N(z) =  z^N \s(z) +  \sum_{m=1}^{N}  \Gamma_{m}  z^{N-m},$$
where
$$ \Gamma_{m} = \begin{pmatrix} \alpha_{1;m} & \beta_{1;m} & \beta_{3;m} \\
\gamma_{1;m} & \alpha_{2;m} - \alpha_{1;m} &  \beta_{2;m} \\
\gamma_{3;m} & \gamma_{2;m} & - \alpha_{2;m}
\end{pmatrix}.$$

Note, that  $\widetilde{\Lmatr}_N(z)$ can be obtained from the $\r$-matrix
$\zeta^N\r(z,\zeta)$.

\subsection{Spectral curve}
The spectral curve of hamiltonian systems in $\mathcal{M}_N$
is defined by the characteristic polynomial of $\widetilde{\Lmatr}_N(z)$. Namely
\begin{equation}\label{SpecC}
\begin{split}
 0 &= \det \big(\widetilde{\Lmatr}_N(z) - w\big) = -w^3 + w \mathcal{I}_{2N-1}(z) + \mathcal{I}_{3N+1}(z)  \\
 &= -w^3  + z^{3N+1} + w \sum_{k=1}^{2N} h_{3k+2} z^{2N-k}
 + \sum_{k=1}^{3N} h_{3k+3} z^{3N-k},
\end{split}
\end{equation}
where parameters $h_\kappa$ of the curve serve as integrals of motion. Actually,
\begin{equation}\label{InvDef}
\begin{split}
&\textstyle \mathcal{I}_{2N-1}(z) \equiv  \sum_{k=0}^{2N-1} h_{6N+2-3k} z^k
= \tfrac{1}{2} \tr \widetilde{\Lmatr}_N(z)^2,\\
&\textstyle \mathcal{I}_{3N+1}(z) \equiv
z^{3N+1} + \sum_{k=0}^{3N-1} h_{9N+3-3k} z^k =
\tfrac{1}{3} \tr \widetilde{\Lmatr}_N(z)^3.
\end{split}
\end{equation}

\begin{prop}
In $\M_N$ defined by \eqref{BousMNSp}, $\dim \M_N \,{=}\, 8N$, there exist $2N$ Casimir functions
which annihilate the Poisson bracket \eqref{LiePoiBraBous}, namely
\begin{equation}\label{hConstr}
\begin{split}
&h_{3k+2} = \tfrac{1}{2} \sum_{\substack{i_1+i_2 =2N-k \\ i_1,i_2 \geqslant 0} }
\tr \big( \Gamma_{N-i_1}\Gamma_{N-i_2}\big),\quad k\in \overline{N+1,2N}\\
&h_{3k+3} = \tfrac{1}{3} \sum_{\substack{i_1+i_2+i_3=3N-k \\ i_1,i_2,i_3 \geqslant 0} }
\tr \big(\Gamma_{N-i_1}\Gamma_{N-i_2} \Gamma_{N-i_3}\big),\quad k\in \overline{2N+1,3N}.
\end{split}
\end{equation}
\end{prop}
\begin{proof}
The statement is proven by  straightforward computations. Namely, with
the Poisson structure  defined  given by (\ref{LiePoiBraBous}) we have
\begin{equation*}
 \sum_{\ell =0}^{N-1}  \sum_{a =1}^8
W_{\ell,\ell'}^{a,b} \frac{\partial h_{\kappa}}{\partial L_{a;\ell}} = 0,\quad
b \in \overline{1,8},\  \ell' \in \overline{1,N}.
\end{equation*}
for all $h_\kappa$ listed in \eqref{hConstr}.
\end{proof}

The $2N$ equations \eqref{hConstr} serve as constraints
 in $\mathcal{M}_N$. By fixing values of $h_{3k+2}$, $k \,{\in}\, \overline{N+1,2N}$, and
$h_{3k+3}$, $k\,{\in}\, \overline{2N+1,3N}$,
 an orbit $\mathcal{O} \,{\subset }\,\mathcal{M}_N$, $\dim \mathcal{O} \,{=}\, 6N$,
 is determined, which serves as the phase space
of a hamiltonian system in $\mathcal{M}_N$.
The remaining $3N$ parameters
$h_{3k+2}$, $k\in \overline{1,N}$, and
$h_{3k+3}$, $k\in \overline{1,2N}$, give rise to non-trivial flows,
which we call hamiltonians.

\begin{prop}
Each hamiltonian system in $\mathcal{M}_N$ has $3N$ degrees of freedom,
and possesses $3N$  hamiltonians, and so
 is integrable in the  sense of Liouville.
\end{prop}

\subsection{Boussinesq equation}
\begin{theo}
The flows of  $h_5$, $h_6$ generate the Boussinesq equation.
\end{theo}
\begin{proof}
The hamiltonian $h_5$ gives rise to a stationary flow parametrized by $\x$,
and $h_6$ gives rise to an evolutionary flow with parameter $\rmt$:
\begin{gather}\label{TwoFlowsEq}
\partial_{\rmx} L_{a;\ell} = \{h_5, L_{a;\ell}\},\qquad\qquad
\partial_{\rmt} L_{a;\ell} = \{h_6, L_{a;\ell}\},
\end{gather}
$a\in\overline{1,8},$ $\ell\in\overline{0,N-1}$.
From the stationary flow we use the equations
\begin{subequations}
\begin{align}
&\partial_{\rmx} \alpha_{1;1} = \gamma_{1;1} + 3 \beta_{2;1} \beta_{3;1} - \beta_{3;2},\\
&\partial_{\rmx} \beta_{1;1} = \alpha_{2;1} - 2 \alpha_{1;1} + 3 \beta_{3;1}^2, \label{DxBeta1} \\
&\partial_{\rmx} \beta_{2;1} = \alpha_{1;1} - 2 \alpha_{2;1}, \label{DxBeta2}\\
&\partial_{\rmx} \beta_{3;1} = \beta_{2;1} - \beta_{1;1}, \label{DxBeta3}
\intertext{and from the evolutionary flow}
&\partial_{\rmt} \beta_{2;1} = -\gamma_{1;1} -  \beta_{3;1} (2 \beta_{1;1} + \beta_{2;1}) + \beta_{3;2},\\
&\partial_{\rmt} \beta_{3;1} = -\alpha_{1;1} -\alpha_{2;1} + 3 \beta_{3;1}^2.
\end{align}
\end{subequations}
By staightforward computations we find
\begin{subequations}
\begin{align}
&\partial_{\rmt} \beta_{3;1}  = \partial_{\rmx} \beta_{1;1}  + \partial_{\rmx} \beta_{2;1}, \label{PreBousEq1}\\
&\partial_{\rmt} \beta_{2;1} = - \partial_{\rmx} \alpha_{1;1} + 2  \beta_{3;1} \partial_{\rmx} \beta_{3;1}. \label{PreBousEq2}
\end{align}
\end{subequations}
 We eliminate $\partial_{\rmx} \beta_{1;1}$ from \eqref{PreBousEq1},
 using the derivative  of \eqref{DxBeta3} with respect to $\rmx$, namely
\begin{equation}\label{DxBeta1Sub}
\partial_{\rmx} \beta_{1;1} = \partial_{\rmx} \beta_{2;1} - \partial^2_{\rmx} \beta_{3;1}.
\end{equation}
Then, from the equation ${-}\tfrac{2}{3}$\eqref{DxBeta1}$-\tfrac{1}{3}$\eqref{DxBeta2}, where
 $\partial_{\rmx} \beta_{1;1}$ is removed by means of \eqref{DxBeta1Sub},
we find
\begin{equation}
\alpha_{1;1} = - \partial_{\rmx} \beta_{2;1} + \tfrac{2}{3} \partial^2_{\rmx} \beta_{3;1},
+  \beta_{3;1}^2,
\end{equation}
and so eliminate $\partial_{\rmx} \alpha_{1;1}$ from \eqref{PreBousEq2}.

Finally, we obtain the system of equations
\begin{subequations}
\begin{align}
&\partial_{\rmt} \beta_{3;1}  = - \partial_{\rmx}^2 \beta_{3;1} + 2 \partial_{\rmx} \beta_{2;1}, \label{BousEq1}\\
&\partial_{\rmt} \beta_{2;1} = -\tfrac{2}{3} \partial^3_{\rmx} \beta_{3;1}  - 2 \beta_{3;1} \partial_{\rmx} \beta_{3;1}
+ \partial_{\rmx}^2 \beta_{2;1}, \label{BousEq2}
\end{align}
\end{subequations}
which, after the substitution
\begin{equation*}
\beta_{3;1}  = \tfrac{1}{3} \rmw,\qquad
\beta_{2;1} =  \tfrac{1}{3} \rmv + \tfrac{1}{6} \rmw_{\rmx},
\end{equation*}
turns into
\begin{equation}\label{BousSystHier}
\begin{split}
&\rmw_{\rmt} = 2 \rmv_{\rmx}, \\
&\rmv_{\rmt}  = - \tfrac{2}{3} \rmw \rmw_{\rmx}  - \tfrac{1}{6} \rmw_{\rmx\rmx\rmx},
\end{split}
\end{equation}
and then into the  Boussinesq equation
\begin{equation}\label{BousEqHier}
3 \rmw_{\rmt\rmt} + 4 \rmw \rmw_{\rmx\rmx} + 4 \rmw_{\rmx}^2 + \rmw_{\rmx\rmx\rmx\rmx} = 0.
\end{equation}
\end{proof}

\subsection{Zero curvature representation}
The system of dynamical equations \eqref{TwoFlowsEq} admits  the matrix form
\begin{gather}\label{MatrGradEqs}
\partial_{\rmx}  \widetilde{\Lmatr}_N(z) = [\nabla h_{5}, \widetilde{\Lmatr}_N(z)], \qquad\quad
\partial_{\rmt} \widetilde{\Lmatr}_N(z) = [\nabla h_{6}, \widetilde{\Lmatr}_N(z)],
\end{gather}
where $\nabla h$ denotes the matrix gradient of $h$, namely,
\begin{gather*}
\nabla h = \sum_{\ell =0}^{N-1}  \sum_{a =1}^8
\frac{\partial h}{\partial L_{a;\ell}} \X_{a,\ell}.
\end{gather*}
So we find
\begin{gather}
\begin{split}
&\nabla h_{5} =  \begin{pmatrix}
0 & 1 & 0 \\ 0 & 0 & 1 \\ z - 3\beta_{2;1} &  - 3\beta_{3;1} & 0
 \end{pmatrix},\\
&\nabla h_{6} = \begin{pmatrix}
2\beta_{3;1}  & 0 & 1 \\
z - 2 \beta_{1;1} - \beta_{2;1} & -\beta_{3;1} & 0 \\
3 \alpha_{1;1} - 6 \beta_{3;1}^2 & z - \beta_{1;1} - 2\beta_{2;1} & -\beta_{3;1}
 \end{pmatrix}.
 \end{split} \notag
\end{gather}
The zero curvature representation for the Boussinesq hierarchy has the form
\begin{gather*}
\partial_{\rmt} \nabla h_{5} - \partial_{\rmx} \nabla h_{6}
+ [\nabla h_{5}, \nabla h_{6} ] = 0.
\end{gather*}

\section{Separation of variables}\label{s:SoV}
\begin{theo}\label{T:SepVars}
Let the phase space $\mathcal{O}$, $\dim \mathcal{O} \,{=}\, 6N$,
of a hamiltonian system of the Boussinesq hierarchy
be parametrized by the dynamic variables
$\{\alpha_{1;m}$, $\alpha_{2;m}$, $\beta_{1;m}$, $\beta_{2;m}$, $\beta_{3;m}$,
$\gamma_{1;m} \mid m\in \overline{1,N}\}$ which satisfy the constraints \eqref{hConstr}.
Then the points $\{(z_k,w_k)\}_{k=1}^{3N}$ which form the divisor of zeros of the system
\begin{equation}\label{SoVInDynVarsSymb}
\begin{split}
&w \big(z^N + \beta_{2}(z) \big) + \mathrm{\it B}_{2}(z)  = 0,\\
& w \beta_{3}(z) + \big(z^{2N} + \mathrm{\it B}_{3}(z) \big) = 0
\end{split}
\end{equation}
belong to the spectral curve \eqref{SpecC}, and form a non-special divisor.
The polynomials $\mathrm{\it B}_{2}(z)$, $\mathrm{\it B}_{3}(z)$ are defined by \eqref{BDefs}.
\end{theo}
A proof is made by the method proposed in  \cite{BerHol07}.

The polynomials $\mathcal{I}_{2N-1}(z)$, $\mathcal{I}_{3N+1}(z)$ defined by \eqref{InvDef},
as functions of the dynamic variables,
can be represented as follows
\begin{gather}\label{IMatr}
\begin{pmatrix}
\mathcal{I}_{2N-1}(z) \\
\mathcal{I}_{3N+1}(z)
\end{pmatrix} =
\begin{pmatrix}
z^{N} + \beta_2(z) & \beta_3(z) \\
\mathrm{\it B}_{2}(z) & z^{2N} + \mathrm{\it B}_{3}(z)
\end{pmatrix}
\begin{pmatrix}
\gamma_2 (z) \\
\gamma_3 (z)
\end{pmatrix}   +
\begin{pmatrix}
\mathrm{\it A}_{2}(z) \\
\mathrm{\it A}_{3}(z)
\end{pmatrix},
\end{gather}
where
\begin{subequations}\label{ABDefs}
\begin{align}
\begin{split}
&\mathrm{\it B}_{2}(z) =
- \begin{vmatrix} \alpha_1(z) & \beta_3(z)\\
\beta_{3;1} z^{N} + \gamma_1(z) & z^{N} + \beta_2(z)
  \end{vmatrix},\\
&\mathrm{\it B}_{3}(z) =
\begin{vmatrix}
z^N + \beta_1(z) & \beta_3(z) \\
\alpha_2(z) - \alpha_1(z) & z^N + \beta_2(z)
\end{vmatrix} -z^{2N},
\end{split}   \label{BDefs}  \\
\begin{split}
&\mathrm{\it A}_{2}(z) =
 \alpha_2(z)^2  - \begin{vmatrix} \alpha_1(z) & z^N + \beta_1(z)\\
\beta_{3;1} z^{N} + \gamma_1(z) & \alpha_2(z) - \alpha_1(z)
  \end{vmatrix} \\
&\qquad\qquad    - 2 \beta_{3;1} z^N \big(z^N + \beta_2(z)\big)
+  \big(z^{N+1} - (\beta_{1;1} + \beta_{2;1})z^N)\big) \beta_3(z),\\
&\mathrm{\it A}_{3}(z) =
-\alpha_2(z) \begin{vmatrix} \alpha_1(z) & z^N + \beta_1(z)\\
\beta_{3;1} z^{N} + \gamma_1(z) & \alpha_2(z) - \alpha_1(z)
  \end{vmatrix} \\
&\qquad\qquad   - 2 \beta_{3;1} z^N  \mathrm{\it B}_{2}(z)
+ \big(z^{N+1} - (\beta_{1;1} + \beta_{2;1})z^N\big) \big(z^{2N} + \mathrm{\it B}_{3}(z) \big).
\end{split} \label{ADefs}
\end{align}
\end{subequations}

\begin{prop}
Each point of the divisor $\{(z_k,w_k)\}_{k=1}^{3N}$ defined by
\eqref{SoVInDynVarsSymb} satisfies the equation
\begin{equation}\label{CIfB0}
- w^3 + w \mathrm{\it A}_{2}(z) + \mathrm{\it A}_{3}(z) = 0,
\end{equation}
or in the factorized form
\begin{equation}\label{CIfB0Factor}
 \big(w + \alpha_{2}(z)\big) \bigg({-}w^2 + w \alpha_{2}(z)
- \begin{vmatrix} \alpha_1(z) & z^N + \beta_1(z)\\
\beta_{3;1} z^{N} + \gamma_1(z) & \alpha_2(z) - \alpha_1(z)
  \end{vmatrix}\bigg) = 0,
\end{equation}
which is equivalent to  the spectral curve equation \eqref{SpecC}
under the conditions \eqref{SoVInDynVarsSymb}.
\end{prop}
A proof is based on straightforward computations.

\begin{theo}
Coordinates of the points $\{(z_k,w_k)\}_{k=1}^{3N}$ defined by \eqref{SoVInDynVarsSymb}
serve as variables of separation for the hamiltonian system from Theorem~\ref{T:SepVars},
that is,  these pairs of coordinates are quasi-canonically conjugate  with respect to the  Lie-Poisson
bracket \eqref{LiePoiBraBous}:
\begin{equation}\label{CanonCoord}
\{z_k, z_l\} = 0,\qquad  \{z_k, w_l\} = z_k^N \delta_{kl},\qquad
\{w_k, w_l\} = 0.
\end{equation}
\end{theo}
A proof is based on the following lemmas, and repeats the exposition
in \cite[pp.\,921--923]{BerHol07}. 
Alternative proofs for the lemmas can be found in \cite{SkrJMP2018,SkrJMP2021}.

\begin{lem}
The divisor $\{(z_k,w_k)\}_{k=1}^{3N}$ defined by \eqref{SoVInDynVarsSymb}
 is alternatively defined by the system
\begin{equation}\label{ABEqs}
\mathcal{B}(z) = 0,\qquad w = \mathcal{A}(z),
\end{equation}
where
\begin{gather}\label{ABbous}
\mathcal{B}(z) = \begin{vmatrix}
z^{N} + \widetilde{\beta}_2(z) & \widetilde{\beta}_3(z) \\
\widetilde{\mathrm{\it B}}_{2}(z) & z^{2N} + \widetilde{\mathrm{\it B}}_{3}(z)
\end{vmatrix}, \qquad
\mathcal{A}(z) =  - \frac{z^{2N}+\widetilde{\mathrm{\it B}}_{3}(z)}{\widetilde{\beta}_3(z)}.
\end{gather}
\end{lem}

\begin{lem}
$\mathcal{A}$ and $\mathcal{B}$ defined by \eqref{ABbous}
satisfy the following equalities with respect to the Lie-Poisson bracket \eqref{LiePoiBraBous}
\begin{equation}\label{ABCommut}
\begin{split}
&\{\mathcal{B}(z), \mathcal{B}(\zeta)\} = 0,\qquad
\{\mathcal{A}(z), \mathcal{A}(\zeta)\} = 0,\\
&\{\mathcal{A}(z), \mathcal{B}(\zeta)\} = - \frac{z^N}{z-\zeta} \mathcal{B}(\zeta) +
 \frac{\zeta^N}{z-\zeta} \mathcal{B}(z) \frac{\beta_3(\zeta)^2}{\beta_3(z)^2}.
\end{split}
\end{equation}
\end{lem}



\begin{lem}
Let $\mathcal{A}$ and $\mathcal{B}$ satisfy \eqref{ABCommut}, then
the variables $\{(z_k,w_k)\}_{k=1}^{3N}$ defined by \eqref{ABEqs}
are quasi-canoniacally conjugate with respect to the Lie-Poisson bracket \eqref{LiePoiBraBous},
namely
\begin{equation}\label{QuCanonVoS}
\{z_k,z_l\} = 0,\qquad \{w_k,w_l\} = 0, \qquad
\{z_k, w_l\} = z_k^N \delta_{kl}.
\end{equation}
\end{lem}

\section{Algebro-geometric integration}\label{s:AGI}
\subsection{$(3,3N+1)$-Curves}\label{ss:HC}
The spectral curves \eqref{SpecC} of the Boussinesq hierarchy form
the family of $(3,3N+1)$-curves with $h_2=0$, namely
\begin{equation}\label{C33m1}
\begin{split}
\mathcal{V}:\quad
f(z,w;\lambda) &\equiv -w^3 + w \mathcal{I}_{2N-1}(z) + \mathcal{I}_{3N+1}(z) \\
&= -w^3  + z^{3N+1} + w \sum_{k=1}^{2N} h_{3k+2} z^{2N-k}
 + \sum_{k=1}^{3N} h_{3k+3} z^{3N-k},
 \end{split}
\end{equation}
which are trigonal curves of genera $g \,{=}\, 3N$, $N \,{\in}\, \Natural$. 
The curves belong to the class of $(n,s)$-curves, known as canonical forms of
plane algebraic curves. The theory of uniformazation of canonical curves respects
the Sato weight: $\wgt z \,{=}\, 3$, $\wgt w \,{=}\, 3N \,{+}\,1$, $\wgt h_{\kappa} \,{=}\, \kappa$.
The point at infinity is a Weierstrass point, and the branch point where all three sheets wind;
it serves as the basepoint of the Abel map.

We assume that the winding numbers of all finite branch points equal one. That is,
parameters $h = (h_{\kappa})$ of the curve $\mathcal{V}$ belong to $\Complex^{5N}\backslash \Discr$,
where $\Discr$ denotes the manifold formed by  $h$ such that the genus of $f(z,w;h)=0$ is less than\;$3N$.
The Weierstrass gap sequence of $\mathcal{V}$ is
\begin{equation*}
\mathfrak{W} = \{\mathfrak{w}_i\}_{i=1}^{g} 
= \Ord \Big(\{3i-2 \mid i=1,\,\dots,\,N \} \cup \{3i-1 \mid i=1,\,\dots,\,2N \} \Big),
\end{equation*}
where $\Ord$ denotes the operator of ordering ascendingly. Let $\mathfrak{M}$ be the
 list of monomials $\mathcal{M}_{3i+(3N+1)j} \,{=}\, z^i w^j$ ordered by their weights $3i+(3N+1)j$, that is
\begin{equation}\label{MList}
\begin{split}
\mathfrak{M} = \big\{1,\,x,\,\dots,\, &x^{N},\, y,\, x^{N+1},\, x y,\,  \dots,\\
 &x^{2N},\, x^{N} y,\, y^2,\,
  \{ x^{2N+i},\, x^{N+i} y,\, x^i y^2\}_{i\in \Natural} \big\}.
\end{split}
\end{equation}

Not normalized differentials of the first kind 
$\rmd u = (\rmd u_{\mathfrak{w}_1}$, $\rmd u_{\mathfrak{w}_2}$, \ldots, $\rmd u_{\mathfrak{w}_{g}})^t$ are 
\begin{equation}\label{K1DifsGen} 
\begin{split}
& \rmd u_{3i-2} =  \frac{w z^{N-i} \rmd z}{\partial_w f(z,w;h)},\quad i=1,\dots, N,\\ 
& \rmd u_{3i-1} =  \frac{z^{2N-i} \rmd z}{\partial_w f(z,w;h)},\quad i=1,\dots, 2N,
\end{split}
\end{equation}
where $\partial_w f(z,w;h) \,{=}\, {-}3w^2 \,{+}\,  \mathcal{I}_{2N-1}(z) $.
Differentials of the second kind $\rmd r = (\rmd r_{\mathfrak{w}_1}$, $\rmd r_{\mathfrak{w}_2}$, 
\ldots, $\rmd r_{\mathfrak{w}_{g}})^t$ 
associated with $\rmd u$ are constructed according to  \cite[\S\,138]{bakerAF}.
Note, that $\wgt \rmd u_{\mathfrak{w}_i} \,{=}\, {-} \mathfrak{w}_i$, and
$\wgt \rmd r_{\mathfrak{w}_i} \,{=}\, \mathfrak{w}_i$. That is, the only pole of
$\rmd r_{\mathfrak{w}_i} $ is located at infinity and has the order  $\mathfrak{w}_i$.

First and second kind periods along the canonical cycles 
$\mathfrak{a}_k$, $\mathfrak{b}_k$, $k\,{=}\,1$, \ldots, $g$, 
are defined as follows
\begin{gather}\label{NNormPerMatr1}
  \omega_k = \oint_{\mathfrak{a}_k} \rmd u,\qquad\qquad
  \omega'_k = \oint_{\mathfrak{b}_k} \rmd u,\\
  \eta_k = \oint_{\mathfrak{a}_k} \rmd r,\qquad\qquad
  \eta'_k = \oint_{\mathfrak{b}_k} \rmd r.
\end{gather}
The vectors $\omega_k$, $\omega'_k$ 
form  first kind period matrices $\omega$, $\omega'$, respectively.
Similarly, $\eta_k$, $\eta'_k$ form second kind period matrices $\eta$, $\eta'$.

The corresponding normalized period matrices of the first kind are $\Ibb_g$, $\tau$,
where $\Ibb_g$ denotes the identity matrix of size $g$, 
and $\tau = \omega^{-1}\omega'$.  The matrix $\tau$ belongs to the Siegel upper half-space,  
that is $\tau^t=\tau$, $\ImN \tau >0$.

\subsection{Abel's map and entire functions}
Let $\{\omega, \omega'\}$ be the period lattice generated from the vectors $\omega_k$, $\omega'_k$.
Then $\Jac(\mathcal{V}) \,{=}\, \Complex^g/\{\omega, \omega'\}$ is the Jacobian variety of $\mathcal{V}$.
Let $u=(u_{\mathfrak{w}_1}$, $u_{\mathfrak{w}_2}$, \ldots, $u_{\mathfrak{w}_g})^t$ denote a  point of $\Jac(\mathcal{V})$.

The Abel map on $\mathcal{V}$, and $\mathcal{V}^n$ are defined by
\begin{align*}
 &\mathfrak{A}(P) = \int_{\infty}^P \rmd u,\qquad P=(z,w)\in \mathcal{V},\\
 &\mathfrak{A}(D) = \sum_{i =1}^n \mathfrak{A}(P_i),\quad D = \sum_{i =1}^n P_i.
\end{align*}
The map is one-to-one on the $g$-th symmetric power of the curve:
 $\mathfrak{A}: \mathcal{V}^g \mapsto \Jac(\mathcal{V})$.

The Riemann theta function  is defined by 
\begin{gather}\label{ThetaDef}
 \theta(v;\tau) = \sum_{n\in \Integer^g} \exp \big(\imath \pi n^t \tau n + 2\imath \pi n^t v\big).
\end{gather}
In what follows, the $\theta$-function is  related to the curve \eqref{C33m1}, that is 
 $v \,{=}\, \omega^{-1}u$, $u \,{\in}\, \Jac(\mathcal{V})$, and  $\tau \,{=}\ \omega^{-1}\omega'$.
 The $\theta$-function with characteristic $[\varepsilon] \,{=}\, (\varepsilon', \varepsilon)^t$ is
\begin{equation}\label{ThetaDefChar}
 \theta[\varepsilon](v;\tau) = \exp\big(\imath \pi \varepsilon'{}^t \tau \varepsilon'
 + 2 \imath \pi  (v+\varepsilon)^t \varepsilon'\big)  \theta(v+\varepsilon + \tau \varepsilon';\tau).
\end{equation}
A characteristic $[\varepsilon]$ is a $2\,{\times}\, g$ matrix, all components of $\varepsilon$, and $\varepsilon'$
are real values within the interval $[0,1)$. Each characteristic represents a point 
in the fundamental domain of $\Jac(\mathcal{V})$, namely
\begin{equation}\label{uCharDef}
u[\varepsilon] = \omega \varepsilon + \omega' \varepsilon'.
\end{equation}

The modular invariant entire function on $\Complex^g \,{\supset}\, \Jac(\mathcal{V})$ is called the sigma function,
which we define after \cite[Eq.(2.3)]{belHKF}:
\begin{equation}\label{SigmaThetaRel}
\sigma(u) = C \exp\big({-}\tfrac{1}{2} u^t \varkappa u\big) \theta[K](\omega^{-1} u;  \omega^{-1} \omega'),
\end{equation}
where  $\varkappa \,{=}\, \eta \omega^{-1}$ is a symmetric matrix, and $[K]$
denotes the characteristic of the vector of Riemann constants.

 \subsection{Uniformization of the spectral curve}
Uniformization is realized through a solution to the  Jacobi inversion problem,
which is expressed in terms of the multiply periodic functions
\begin{gather*}
\wp_{i,j}(u) = -\frac{\partial^2 \log \sigma(u) }{\partial u_i \partial u_j },\qquad
\wp_{i,j,k}(u) = -\frac{\partial^3 \log \sigma(u) }{\partial u_i \partial u_j \partial u_k}.
\end{gather*}

\begin{theo}[\cite{BLJIP22}, Theorem 3]
Let $u = \mathfrak{A}(D )$ be the Abel image of  a non-special positive divisor  
$D \in \mathcal{V}^{3N}$ on a $(3,3N+1)$-curve $\mathcal{V}$ defined by \eqref{C33m1}. 
Then $D$ is the common divisor of zeroes of the two functions
\begin{subequations}\label{JIPC33m1}
\begin{align}
&\mathcal{R}_{6N}(z,w;u) \equiv z^{2N} + \textstyle \sum_{i=1}^{N} p_{3i-1} w z^{N-i}
+ \sum_{i=1}^{2N} p_{3i} z^{2N-i}, \label{R2g}\\ 
&\mathcal{R}_{6N+1}(x,y;u) \equiv w z^N 
+ \textstyle \sum_{i=1}^{N} q_{3i} w z^{N-i} 
+ \sum_{i=1}^{2N} q_{3i+1} z^{2N-i}, \label{R2g1}
\end{align}
where
\begin{equation}\label{BasisFunct}
p_{\wFr_i+1} = - \wp_{1,\wFr_i}(u),\quad
q_{\wFr_i+2} = \tfrac{1}{2} \big(\wp_{1,1,\wFr_i}(u) - \wp_{2,\wFr_i}(u) \big),
\quad  \wFr_i \in \mathfrak{W}.
\end{equation}
\end{subequations}
\end{theo}
The functions \eqref{BasisFunct} form a convenient basis 
in the abelian function field associated with the curve $\mathcal{V}$, see \cite{BerWPFF2025}.
In fact, every meromorphic function on $\Jac(\mathcal{V})\backslash \Sigma$ 
is represented as a rational function in this basis.

Comparing \eqref{JIPC33m1} with \eqref{SoVInDynVarsSymb}, we immediately find
\begin{equation}
\begin{split}
&\beta_{2;m} = q_{3m} = \tfrac{1}{2} \big(\wp_{1,1,3m-2}(u) - \wp_{2,3m-2}(u) \big),\\
&\beta_{3;m} = p_{3m-1} = - \wp_{1,3m-2}(u),\\
&\mathrm{\it B}_{2;m} = q_{3m+1} = \tfrac{1}{2} \big(\wp_{1,1,3m-1}(u) - \wp_{2,3m-1}(u) \big),\\
&\mathrm{\it B}_{3;m} = p_{3m}  =  - \wp_{1,3m-1}(u).
\end{split}
\end{equation}
The obtained equalities are solvable for the dynamic variables 
$\{\alpha_{1;m}$, $\alpha_{2;m}$, $\beta_{1;m}$, $\beta_{2;m}$, $\beta_{3;m}$,
$\gamma_{1;m} \mid m\in \overline{1,N}\}$, which describe a finite-gap
hamiltonian system in the Boussinesq hierarchy, see Theorem~\ref{T:SepVars}.

As seen from \eqref{BousSystHier}, 
the main variables $\rmw$, $\rmv$  of the Boussinesq hierarchy 
are expressed in terms of $\beta_{2;1}$,  $\beta_{3;1}$ as follows
\begin{equation*}
\begin{split}
&\rmw(\rmx, \rmt) = 3\beta_{3;1} = - 3 \wp_{1,1}(u), \\
&\rmv(\rmx, \rmt) = 3\beta_{2;1} - \tfrac{3}{2} \partial_{\rmx} \beta_{3;1}
= \tfrac{3}{2} \big(\wp_{1,1,1}(u) - \wp_{1,2}(u) \big) + \tfrac{1}{2}  \partial_{\rmx} \wp_{1,1}(u).
\end{split}
\end{equation*}

\subsection{Equations of motion for variables of separation}
Now we find the stationary and evolutionary equations of motion for $\mathcal{B}(z)$.
From \eqref{MatrGradEqs},  taking into account \eqref{SoVInDynVarsSymb} and \eqref{CIfB0Factor}, 
we find
\begin{equation}\label{DBDxEqs}
\begin{split}
&\partial_{\rmx} \mathcal{B}(z) = \big(3w^2 - \mathcal{I}_{2N-1}(z)\big) \beta_3(z),\\
&\partial_{\rmt} \mathcal{B}(z) = 3 \beta_{3;1} \mathcal{B}(z) - \big(3w^2 - \mathcal{I}_{2N-1}(z)\big)
\big(z^{N} + \beta_2(z)\big),
\end{split}
\end{equation}
where all dynamic variables are functions of $\x$ and $\rmt$.

On the other hand,
 $\mathcal{B} (z) = \prod_{k=1}^{3N} (z-z_k(\x,\rmt))$, and 
all zeros of $\mathcal{B}(z)$ are  functions of $\x$ and $\rmt$.
Then, for $k\in \overline{1,3N}$
\begin{subequations}
\begin{align*}
&\frac{\rmd}{\rmd \x} \log \mathcal{B} (z)= - \frac{1}{z-z_k}\frac{\rmd z_k}{\rmd \x}
= \big(3w^2 - \mathcal{I}_{2N-1}(z)\big) \frac{ \beta_3(z)}{\mathcal{B} (z)},\\
&\frac{\rmd}{\rmd \rmt} \log \mathcal{B} (z) = - \frac{1}{z-z_k}\frac{\rmd z_k}{\rmd \rmt}
= 3 \beta_{3;1} - \big(3w^2 - \mathcal{I}_{2N-1}(z)\big)
\frac{z^{N} + \beta_2(z)}{\mathcal{B} (z)}.
\end{align*}
\end{subequations}
As $z\to z_k$, $k=1$, \ldots, $N$, we obtain
\begin{gather}\label{DzDtEqs}
\begin{split}
&\frac{\rmd z_k}{\rmd \x}
= \beta_3(z_k) \frac{\big(3w_k^2 - \mathcal{I}_{2N-1}(z_k)\big) }{\prod_{j\neq k}^{3N} (z_k - z_j)},\\
& \frac{\rmd z_k}{\rmd \rmt}
= - \big(z^{N} + \beta_2(z_k) \big) \frac{\big(3w_k^2 - \mathcal{I}_{2N-1}(z_k)\big)}{\prod_{j\neq k}^{3N} (z_k - z_j)}.
\end{split}
\end{gather}

\begin{theo}
In the Boussinesq hierarchy we have
 $u_{1} =  \x + C_1$, $u_2 = \rmt + C_2$, and $u_{\wFr_i} = C_{\wFr_i}$, $i \in \overline{3,g}$, all $C_n$
 are constant. The finite-gap solution to the Boussinesq equation \eqref{BousEqHier}  in 
a $6N$-dimensional phase space ($N \in \Natural$) is
\begin{equation}\label{KdVSolRealCond}
\begin{split}
&\rmw(\rmx,\rmt) = - 3 \wp_{1,1} ( u + \bm{C}),\qquad u = (\rmx, \rmt, 0, \dots)^t,\\
&\rmv(\rmx,\rmt) = 2 \wp_{1,1,1}(u+ \bm{C}) - \tfrac{3}{2}  \wp_{1,2}(u+ \bm{C}).
\end{split}
\end{equation}
where $\bm{C} = (C_{\wFr_1},\dots,\,C_{\wFr_g})^t$ is a constant vector.
\end{theo}
\begin{proof}
Let $D$ be a divisor of points $\{(z_k,w_k)\}_{k=1}^{3N}$ defined by \eqref{SoVInDynVarsSymb}.
Combining \eqref{K1DifsGen} with \eqref{DzDtEqs}, we find
\begin{subequations}
\begin{align*}
&\frac{\rmd u_{3n-1}}{\rmd \x} = \sum_{k=1}^{3N} \frac{z_k^{2N-n}}{3w_k^2 - \mathcal{I}_{2N-1}(z_k)} 
\frac{\rmd z_k}{\rmd \x}
= - \sum_{k=1}^{3N} \frac{z_k^{2N-n} \beta_3(z_k)}{\prod_{j\neq k}^{3N} (z_k - z_j)} = 0,\
n\in \overline{1,2N}, \\
&\frac{\rmd u_{3n-2}}{\rmd \x} = \sum_{k=1}^{3N} \frac{w_k z_k^{N-n}}
{3w_k^2 - \mathcal{I}_{2N-1}(z_k)} \frac{\rmd z_k}{\rmd \x}
= \sum_{k=1}^{3N} \frac{z_k^{N-n} \big(z^{2N} + {\it B}_3(z_k)\big)}{\prod_{j\neq k}^{3N} (z_k - z_j)}
= \delta_{n,1},\ n\in \overline{1,N}, \\
& \frac{\rmd u_{3n-1}}{\rmd \rmt} = \sum_{k=1}^{3N} \frac{z_k^{2N-n}}{3w_k^2 - \mathcal{I}_{2N-1}(z_k)}
 \frac{\rmd z_k}{\rmd \rmt}
= \sum_{k=1}^{3N} \frac{z_k^{2N-n} \big(z^{N} + \beta_2(z_k) \big)}{\prod_{j\neq k}^{3N} (z_k - z_j)}
= \delta_{n,1},\ n\in \overline{1,2N}, \\
& \frac{\rmd u_{3n-2}}{\rmd \rmt} = \sum_{k=1}^{3N} \frac{w_k z_k^{N-n}}{3w_k^2 - \mathcal{I}_{2N-1}(z_k)}
 \frac{\rmd z_k}{\rmd \rmt}
=  - \sum_{k=1}^{3N} \frac{z_k^{N-n} {\it B}_2(z_k)}{\prod_{j\neq k}^{3N} (z_k - z_j)} = 0,\ 
n\in\overline{1,N}.
\end{align*}
\end{subequations}
\end{proof}

\begin{rem}\label{r:CompKdV}
The obtained finite-gap solution of the Boussinesq equation 
is given by the function $\wp_{1,1}$
associated with a trigonal curve,
which is similar to the finite-gap solution of the KdV equation,
given by the same function associated with a hyperelliptic curve.
\end{rem}

\begin{rem}
The Boussinesq equation \eqref{BousEqHier} arises as a dynamical equation
for $\wp_{1,1}$ on $\Jac(\mathcal{V}) \backslash \Sigma$, namely
\begin{equation}\label{BoussWP}
- 3 \wp_{1,1,2,2} + 12 \wp_{1,1} \wp_{1,1,1,1} + 12 \wp_{1,1,1}^2 - \wp_{1,1,1,1,1,1} = 0,
\end{equation}
which is obtained from the identities
\begin{gather}
\begin{split}
& \wp_{1,1,1,1,1,1} =  30 \wp_{1,1,1}^2 -15 \wp_{1,1,2,2} +  60 \wp_{1,1} \wp_{2,2}
- 24 \wp_{1,5} + 30 \wp_{1,2}^2 + 24 h_6,\\
&\wp_{1,1,1,1} = 6 \wp_{1,1}^2 - 3 \wp_{2,2},\\
&\wp_{1,1,2,2}  = 2 \wp_{1,1} \wp_{2,2} + 4 \wp_{1,2}^2 + 4 \wp_{1,5} + 2 h_6,\\
&\wp_{1,1,1}^2 = 4 \wp_{1,1}^3 - 4 \wp_{1,1}\wp_{2,2} + \wp_{1,2}^2 + 4\wp_{1,5}.
\end{split}
\end{gather}
See \cite{BerWPFF2025} for the details on obtaining the identities for $\wp$-functions associated with a 
trigonal curve. Note, that $\mathcal{V}$ differs from the canonical $(3,3N+1)$-curve by $h_2=0$.

If the parameter $h_2$ of a $(3,3N+1)$-curve does not vanish,
then we come to the identity
\begin{equation}
- 3 \wp_{1,1,2,2}  + 4 h_2 \wp_{1,1,1,1} + 12 \wp_{1,1} \wp_{1,1,1,1} + 12 \wp_{1,1,1}^2 - \wp_{1,1,1,1,1,1} = 0,
\end{equation}
which contains all terms of the original equation \eqref{BousEqNorm}.
\end{rem}

\begin{rem}
The Hirota bilinear equation \cite[Eq.\,(3.13)]{Hir1973}
coincides, up to rescaling, with the bilinear relation 
\begin{equation*}
\mathcal{D}_2^2 - \tfrac{1}{6} h_2 \mathcal{D}_1^2 + \tfrac{1}{4!} \mathcal{D}_1^4 = 0,
\end{equation*}
associated with the $(3,4)$-curve, see \cite[Eq.\,(116)]{BerWPFF2025}, 
and also with the whole family of $(3,3N\,{+}\,1)$-curves.
This bilinear relation represents the Boussinesq equation
in terms of  bilinear operations acting on the $\sigma$-function.
\end{rem}

\section{Reality conditions}
\begin{conj}\label{C:PerRhomb}
Let all finite branch points $\{(e_i,d_i)\}$ of $\mathcal{V}$ split into real, and pairs of complex conjugate $e_i$, $\bar{e}_i$, 
Then period matrices $\omega$ and $\eta$ can be made purely imaginary, and the matrix $\varkappa$  real.
The period lattice is formed by rhombic sublattices, since   $\ImN \omega'_j$ 
is spanned by $\frac{1}{2} \ImN \omega_k$, $k\in \overline{1,g}$, for all $j\in \overline{1,g}$.
\end{conj}

\begin{conj}\label{C:FinSubSp}
Let finite branch points  of $\mathcal{V}$ be real, or complex conjugate. 
\begin{itemize}[labelwidth=1em, leftmargin=!]
\item There exist $2^g$ affine subspaces $\mathfrak{J}^{\ReN} \,{=}\, \{\Omega \,{+}\, s \,{\mid}\, s \,{\in}\, \Real^g\}$,
parallel to the real axes, $\Omega = u[\varepsilon]$ with half-integer 
characteristics $[\varepsilon]$,
such that  $\wp_{i,j}(s\,{+}\,\Omega)$, $\wp_{i,j,k}(s\,{+}\,\Omega)$ are real-valued,
and have poles. 
With the choice of periods as indicated in Conjecture\;\ref{C:PerRhomb},
the corresponding $2^g$ half-periods are purely imaginary: $\Omega \in \imath \Real^g$.
\item On the subspace $\mathfrak{J}^{\ImN} \,{=}\, \{\imath s \,{\mid}\, s\,{\in}\, \Real^g\}$,  
spanned by the imaginary axes, $\wp_{i,j}(\imath s)$  are real-valued, and 
$\wp_{i,j,k}(\imath s)$ acquire purely imaginary values.
\end{itemize}
\end{conj}

\begin{conj}
Let $\bm{C} = u[K]$, then the finite-gap solution \eqref{KdVSolRealCond} is bounded.
\end{conj}

\begin{rem}
Unlike the hyperelliptic case, $\bm{C} \,{=}\, u[K]$ does not belong to any of the subspaces $\mathfrak{J}^{\ReN}$
mentioned in Conjecture\;\ref{C:FinSubSp}. Therefore, $\wp_{i,j}(s\,{+}\, u[K])$, and $\wp_{i,j,k}(s\,{+}\,u[K])$
are complex-valued.
\end{rem}

\begin{rem}
In the case of cyclic curves the same behaviour of solutions  is observed.
\end{rem}

\section{Quasi-periodic solutions}\label{s:SolPlots}
As an example we consider the case of $N\,{=}\,1$.
The phase space $\mathcal{O} \,{\subset}\, \M_1$, $\dim \mathcal{O} \,{=}\, 6$, is described by the dynamic variables 
$\{\alpha_{1;1},\, \alpha_{2;1},\, \beta_{1;1},\, \beta_{2;1},\, \beta_{3;1},\,
\gamma_{1;1}$, $\gamma_{2;1}$, $\gamma_{3;1}\}$ with constraints
\begin{align*}
&\alpha_{1;1}^2 + \alpha_{2;1}^2 - \alpha_{1;1} \alpha_{2;1}
+ \beta_{1;1} \gamma_{1;1} + \beta_{2;1} \gamma_{2;1} + \beta_{3;1} \gamma_{3;1} = h_8,\\
&\alpha_{1;1}^2 \alpha_{2;1} - \alpha_{1;1} \alpha_{2;1}^2
+ \alpha_{1;1} \big(\beta_{3;1} \gamma_{3;1} -  \beta_{2;1} \gamma_{2;1} \big) 
+ \alpha_{2;1} \big( \beta_{1;1} \gamma_{1;1}  - \beta_{3;1} \gamma_{3;1}\big)\\
&\qquad + \beta_{1;1} \beta_{2;1} \gamma_{3;1} + \beta_{3;1} \gamma_{1;1} \gamma_{2;1} = h_{12}.
\end{align*}
The system is governed by the hamiltoninans
\begin{align*}
&h_5 = \gamma_{1;1} + \gamma_{2;1} - 3 \beta_{2;1} \beta_{3;1}, \\
&h_6 = \gamma_{3;1} - \beta_{1;1}^2 -  \beta_{1;1} \beta_{2;1} -  \beta_{2;1}^2 
- 2 \beta_{3;1}^3 + 3 \alpha_{1;1} \beta_{3;1}, \\
&h_9 = - \beta_{1;1}^2 \beta_{2;1} - \beta_{1;1} \beta_{2;1}^2 
+ \beta_{3;1}^2 \big(\gamma_{2;1} - 2 \gamma_{1;1}\big)
- \alpha_{1;1} \big(\gamma_{2;1} + \beta_{1;1} \beta_{3;1} - \beta_{2;1} \beta_{3;1}\big) \\
&\qquad + \alpha_{2;1} \big(\gamma_{1;1} + \beta_{2;1} \beta_{3;1} + 2 \beta_{1;1} \beta_{3;1} \big)
+ \gamma_{3;1} \big(\beta_{1;1} + \beta_{2;1}\big).
\end{align*}

The spectral curve has the form
\begin{equation*}
-w^3 + z^4 + y (h_5 z + h_8) + h_6 z^2 + h_9 z + h_{12} = 0.
\end{equation*}
Associated first and second kind differentials, see \cite[p.\,6]{BerWPFF2025}, are
\begin{gather*}
\begin{split}
&\rmd u = \begin{pmatrix} \rmd u_1 \\ \rmd u_2 \\ \rmd u_5 \end{pmatrix}
= \begin{pmatrix} y \\ x \\ 1 \end{pmatrix} \frac{\rmd z}{-3w^2 + h_5 z + h_8},\\
&\rmd r = \begin{pmatrix} \rmd r_1 \\ \rmd r_2 \\ \rmd r_5 \end{pmatrix}
= \begin{pmatrix} x^2  \\ 2xy \\ 5x^2 y + h_6 y \end{pmatrix} \frac{\rmd z}{-3w^2 + h_5 z + h_8}.
\end{split}
\end{gather*}

We consider several examples of spectral curves and present real-valued solutions
of the Boussinesq equation.

\begin{figure}[ht]
\begin{tabular}{cc}
\parbox[t]{0.4\textwidth}{
\includegraphics[width=0.3\textwidth]{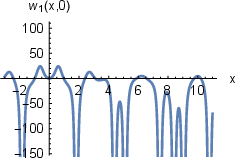}\\
(a) plot of ${-} 3 \wp_{1,1} ( (\rmx,0,0)^t + \Omega_1)$} &
\parbox[t]{0.4\textwidth}{
\includegraphics[width=0.3\textwidth]{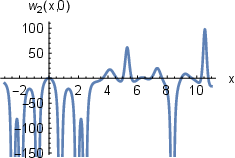} \\ 
(b) plot of $- 3 \wp_{1,1} ( (\rmx,0,0)^t + \Omega_2)$}\\ $\quad$ \\
\parbox[t]{0.4\textwidth}{
\includegraphics[width=0.3\textwidth]{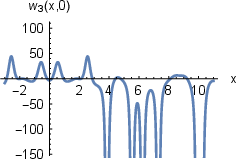}\\ 
(c) plot of ${-} 3 \wp_{1,1} ( (\rmx,0,0)^t + \Omega_3)$} &
\parbox[t]{0.4\textwidth}{
\includegraphics[width=0.3\textwidth]{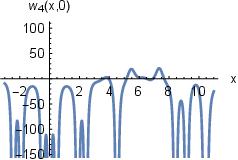}\\ 
(d) plot of $- 3 \wp_{1,1} ( (\rmx,0,0)^t + \Omega_4)$ }\\ $\quad$ \\
\parbox[t]{0.4\textwidth}{
\includegraphics[width=0.3\textwidth]{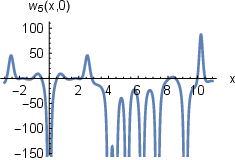}\\
(e) plot of ${-} 3 \wp_{1,1} ( (\rmx,0,0)^t + \Omega_5)$} &
\parbox[t]{0.4\textwidth}{
\includegraphics[width=0.3\textwidth]{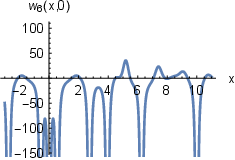}\\ 
(f) plot of $- 3 \wp_{1,1} ( (\rmx,0,0)^t + \Omega_6)$ }\\ $\quad$ \\
\parbox[t]{0.4\textwidth}{
\includegraphics[width=0.3\textwidth]{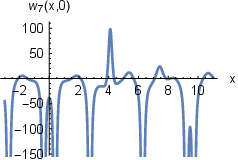}\\
(g) plot of ${-} 3 \wp_{1,1} ( (\rmx,0,0)^t + \Omega_7)$ }&
\parbox[t]{0.45\textwidth}{
\includegraphics[width=0.3\textwidth]{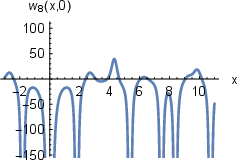}\\
(h) plot of $- 3 \wp_{1,1} ( (\rmx,0,0)^t + \Omega_8)$ }\\ $\quad$ \\
\multicolumn{2}{c}{
\parbox[t]{0.8\textwidth}{
\includegraphics[width=0.58\textwidth]{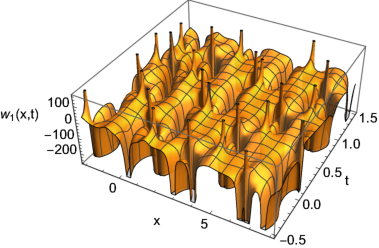} \\
(i) plot of $- 3 \wp_{1,1} ( (\rmx,\rmt,0)^t + \Omega_1)$ }  }
\end{tabular}
\caption{Plots of $\rmw(\rmx,\rmt)  \,{=}\, {-} 3 \wp_{1,1} ( (\rmx,\rmt,0)^t + \bm{C})$
with different values of $\bm{C}$, associated with $\mathcal{V}_{\text{8R}}$.}
\label{wiPlotsEx8R}
\end{figure}
\begin{figure}[ht]
\begin{tabular}{cc}
\parbox[t]{0.4\textwidth}{
\includegraphics[width=0.3\textwidth]{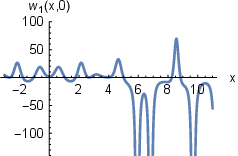}\\
(a) plot of ${-} 3 \wp_{1,1} ( (\rmx,0,0)^t + \Omega_1)$} &
\parbox[t]{0.4\textwidth}{
\includegraphics[width=0.3\textwidth]{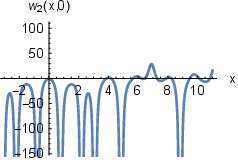} \\ 
(b) plot of $- 3 \wp_{1,1} ( (\rmx,0,0)^t + \Omega_2)$}\\ $\quad$ \\
\parbox[t]{0.4\textwidth}{
\includegraphics[width=0.3\textwidth]{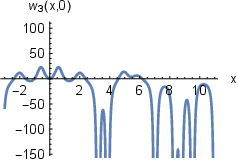}\\ 
(c) plot of ${-} 3 \wp_{1,1} ( (\rmx,0,0)^t + \Omega_3)$} &
\parbox[t]{0.4\textwidth}{
\includegraphics[width=0.3\textwidth]{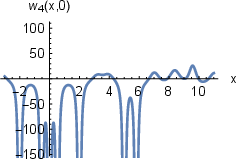}\\ 
(d) plot of $- 3 \wp_{1,1} ( (\rmx,0,0)^t + \Omega_4)$ }\\ $\quad$ \\
\parbox[t]{0.4\textwidth}{
\includegraphics[width=0.3\textwidth]{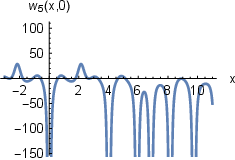}\\
(e) plot of ${-} 3 \wp_{1,1} ( (\rmx,0,0)^t + \Omega_5)$}
&
\parbox[t]{0.4\textwidth}{
\includegraphics[width=0.3\textwidth]{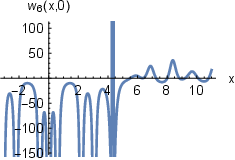}\\ 
(f) plot of $- 3 \wp_{1,1} ( (\rmx,0,0)^t + \Omega_6)$ }\\ $\quad$ \\
\parbox[t]{0.4\textwidth}{
\includegraphics[width=0.3\textwidth]{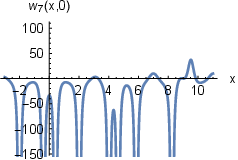}\\
(g) plot of ${-} 3 \wp_{1,1} ( (\rmx,0,0)^t + \Omega_7)$ }&
\parbox[t]{0.45\textwidth}{
\includegraphics[width=0.3\textwidth]{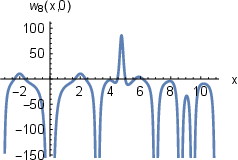}\\
(h) plot of $- 3 \wp_{1,1} ( (\rmx,0,0)^t + \Omega_8)$ }\\ $\quad$ \\
\multicolumn{2}{c}{
\parbox[b]{0.6\textwidth}{
\includegraphics[width=0.58\textwidth]{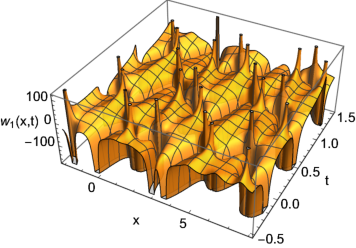} \\
(i) plot of $- 3 \wp_{1,1} ( (\rmx,\rmt,0)^t + \Omega_1)$ } }
\end{tabular}
\caption{Plots of $\rmw(\rmx,\rmt)  \,{=}\, {-} 3 \wp_{1,1} ( (\rmx,\rmt,0)^t + \bm{C})$
with different values of $\bm{C}$, associated with $\mathcal{V}_{\text{2R6C}}$.}
\label{wiPlotsEx2R6C}
\end{figure}

\subsection{The case of all real branch points}\label{ss:8RealBPs}
Let the spectral curve be
\begin{equation*}
\mathcal{V}_{\text{8R}}:\qquad -w^3 + z^4 + y (8 z + 238) - 192 z^2 - 836 z + 680 = 0.
\end{equation*}
Finite branch points $\{(e_i,d_i)\}_{i=1}^8$ are sorted ascendingly by their $z$-coordinates:
\begin{align*}
&e_1 \approx -10.68393,& &e_2 \approx -8.79661,&
&e_3 \approx -8.32493,& \\
&e_4 \approx -4.17746,& &e_5 \approx -1.03382,&
&e_6 \approx 1.87416,& \\
&e_7 \approx 15.25087,& &e_8 \approx 15.89173.&
\end{align*}
The curve has three sheets, denoted by $A$, $B$, $C$.
At each $e_i$ two sheets join, namely
\begin{equation}\label{SheetJ}
\begin{array}{cccccccc}
e_1 & e_2 & e_3 & e_4 & e_5 & e_6 & e_7 & e_8\\
(BC) & (AB) & (AB) & (AC) & (AC) & (BC) & (BC) & (AB)
\end{array}
\end{equation}
Cuts are made along the segments $[e_2,e_3]$,  $[e_4,e_5]$,
 $[e_6,e_7]$, and $[e_8,\infty) \cup  (-\infty,e_1]$.
 $\mathfrak{a}$-Cycles encircle the segments $[e_2,e_3]$,  $[e_4,e_5]$,
 $[e_6,e_7]$ counter-clockwise.  $\mathfrak{b}$-Cycles emerge from the
 cut $[e_8,\infty) \cup  (-\infty,e_1]$, and enter the cuts
  $[e_2,e_3]$,  $[e_4,e_5]$, $[e_6,e_7]$, respectively, without mutual intersections.

The period matrices of the first kind are
\begin{gather*}
\begin{split}
&\omega \approx \begin{pmatrix}
-0.31602 \imath & -0.44719 \imath & 0.46239 \imath \\
0.35885 \imath & -0.13376 \imath & 0.20678 \imath \\
-0.04193 \imath & 0.05026 \imath & 0.0304 \imath
\end{pmatrix},\\
&\omega' \approx \begin{pmatrix}
0.996 + 0.45479 \imath & 1.16657 + 0.6128 \imath & -0.66688 - 0.3816 \imath \\
-0.60095 + 0.17027 \imath & 0.05931 - 0.00915 \imath & -0.21608 + 0.11255 \imath \\
0.05123 - 0.00993 \imath & -0.03771 + 0.01103 \imath & -0.00881 + 0.00417 \imath
\end{pmatrix},
\end{split}
\end{gather*}  
and of the second kind
\begin{gather*}
\begin{split}
&\eta \approx \begin{pmatrix}
-3.07265 \imath & 0.41877 \imath & 2.08359 \imath \\
5.4093 \imath & 2.35886 \imath & 7.03914 \imath \\
-55.10771 \imath & 67.57539 \imath & 91.71058 \imath
\end{pmatrix},\\
&\eta' \approx \begin{pmatrix}
1.29203 + 0.83241 \imath & -3.97836 + 2.36874 \imath & 2.05683 - 1.32694 \imath \\
0.37472 + 2.34014 \imath & -1.16025 - 0.36451 \imath & 6.58485 + 3.88408 \imath \\
61.4774 + 12.0676 \imath & 48.3366 + 39.6215 \imath & 8.27899 + 6.23384 \imath
\end{pmatrix}.
\end{split}
\end{gather*}  
Then, the Riemann period matrix is 
\begin{gather*}
\tau \approx \begin{pmatrix}
1.99808 \imath & -0.5 + 0.70936 \imath & 0.5 -0.10241 \imath \\
-0.5 + 0.70936 \imath & -0.5 + 1.62454 \imath & 0.5 - 0.46696 \imath \\
0.5 - 0.10241 \imath & 0.5 - 0.46696 \imath & 0.92063 \imath
\end{pmatrix},
\end{gather*}  
and the symmetric matrix from the definition of the $\sigma$-function is
\begin{gather*}
\varkappa \approx \begin{pmatrix}
3.25733 & -1.93611 & 32.1638 \\
-1.93611 & 24.4398 & 94.7578 \\
32.1638 & 94.7578 & 1883.02
\end{pmatrix}.
\end{gather*} 
The characteristic $[K]$ of the vector of Riemann constants with such a choice of homology basis is
\begin{equation}
[K] = \begin{pmatrix} \tfrac{1}{2} & \tfrac{1}{2}\vphantom{y_{\int_F R}} & \tfrac{1}{2} \\
0 & \tfrac{1}{2} & 0 \end{pmatrix}.
\end{equation}

In $\Jac(\mathcal{V})$ there exist $8$ subspaces $\mathfrak{J}^{\ReN}$ 
with $\Omega$ generated from
\begin{equation*}
\Omega_1 = u \bigg[\begin{matrix} 0 & 0\vphantom{y_{\int_F R}} & 0 \\
\tfrac{1}{2} & 0 & 0 \ \end{matrix}\bigg],\qquad
\Omega_2 = u \bigg[\begin{matrix} 0 & 0\vphantom{y_{\int_F R}} & 0 \\
0 & \tfrac{1}{2} & 0 \ \end{matrix}\bigg],\qquad
\Omega_3 = u \bigg[\begin{matrix} 0 & 0\vphantom{y_{\int_F R}} & 0 \\
0 & 0 & \tfrac{1}{2} \ \end{matrix}\bigg],
\end{equation*}
where $u[\varepsilon]$ is defined by \eqref{uCharDef}. Let
\begin{equation*}
\Omega_4 = \Omega_1  + \Omega_2,\quad
\Omega_5 = \Omega_1  + \Omega_3,\quad
\Omega_6 = \Omega_2  + \Omega_3,\quad
\Omega_7 = \Omega_1  + \Omega_2 + \Omega_3,\quad
\Omega_8 = 0.
\end{equation*}
Note, that $\ReN \Omega_i = 0$, $i\in \overline{1,8}$.
On fig.\,\ref{wiPlotsEx8R} (a)--(h) plots of 
$w(\rmx,0)$, as defined by \eqref{KdVSolRealCond},
with $\bm{C} = \Omega_i$, $i\in \overline{1,8}$, are presented.
These are all cases where $\rmw$ and $\rmv$ are real-valued.
A plot of $\rmw(\rmx,\rmt)$ with $\bm{C} = \Omega_1$ is shown on fig.\;\ref{wiPlotsEx8R} (i).

\subsection{The case of 6 complex conjugate and 2 real branch points}
We consider the same hamiltonian system with different values of hamiltonians. 
That is, the orbit is defined by the same values: $h_8 = 238$,
and $h_{12}=680$.  Let the spectral curve be defined by
\begin{equation*}
\mathcal{V}_{\text{2R6C}}:\qquad -w^3 + z^4 + y (8 z + 238) - 75 z^2 - 175 z + 680 = 0.
\end{equation*}
$z$-Coordinates of finite branch points $\{(e_i,d_i)\}_{i=1}^8$ are 
\begin{align*}
&e_1 \approx -7.46999,& &e_2 \approx -6.44608 - 2.89683 \imath,&
&e_3 \approx -6.44608 + 2.89683 \imath,& \\
&& &e_4 \approx -1.42284 - 2.73759 \imath,& 
&e_5 \approx -1.42284 + 2.73759 \imath,&\\
&e_8 \approx 10.31227,&
&e_6 \approx 6.44779 - 0.57983 \imath,& 
&e_7 \approx 6.44779 + 0.57983 \imath.& 
\end{align*}
Sheets join as shown in \eqref{SheetJ}.
Cuts are made, and homology basis is chosen in the same way as in the previous example, see 
subsection~\ref{ss:8RealBPs}.

The period matrices used in computation of the $\wp$-functions are
\begin{align*}
&\omega \approx \begin{pmatrix}
0.29976 \imath & 0.58723 \imath & -0.42855 \imath \\
-0.37807 \imath & 0.09173 \imath & -0.28325 \imath \\
0.05712 \imath & -0.07289 \imath & -0.04384 \imath
\end{pmatrix},\\
&\tau \approx \begin{pmatrix}
-0.5 + 1.13252 \imath & 0.5 + 0.63674 \imath & -0.5 -0.09878 \imath \\
0.5 + 0.63674 \imath &  1.33343 \imath & -0.5 - 0.46327 \imath \\
-0.5 -0.09878 \imath& -0.5 - 0.46327 \imath & -0.5 + 1.63202 \imath
\end{pmatrix},\\
&\varkappa \approx \begin{pmatrix}
2.81744 & -0.85404 & 19.5561 \\
-0.85404 & 15.5874 & 33.919 \\
19.5561 & 33.919 & 887.46
\end{pmatrix}.
\end{align*}
The characteristic $[K]$ is the same, since the homology basis is chosen in the same way.
The $8$ subspaces $\mathfrak{J}^{\ReN}$ where the $\wp$-functions are real-valued
are generated by the same characteristics.

Plots of $\rmw$ with  $\bm{C} = \Omega_i$,  $i\in \overline{1,8}$, are presented on
 fig.\,\ref{wiPlotsEx2R6C}.

\section{Conclusion remarks}
The proposed rough construction of the Boussinesq hierarchy shows
that the hierarchy is associated with
a family of trigonal curves of the type $(3,3N\,{+}\,1)$ as spectral curves.
The Boussinesq equation \eqref{BousEqHier} arises as
the dynamical equation \eqref{BoussWP} for every curve from the family.
The finite-gap solution of the Boussinesq equation is given by the function $\wp_{1,1}$,
which coincides with the finite-gap solution of the KdV equation, see Remark~\ref{r:CompKdV}.

The problem of reality conditions requires further investigation,
since only  solutions with singularities have been found.
Obtaining bounded and real-valued solutions is still an open problem,
as well as connection with known solutions of the Boussinesq equation.



\end{document}